%% file: main.tex
\documentclass[runningheads]{llncs}
\usepackage[T1]{fontenc}
\usepackage{graphicx}
\usepackage{amsmath}
\usepackage{amssymb}
\usepackage{mathtools}
\usepackage{booktabs}
\usepackage{makecell} 
\usepackage{array}    

\usepackage{subcaption}

\usepackage{tikz-cd}
\usepackage{tikz}
\usepackage{algorithm}
\usepackage[noend]{algpseudocode}
\usepackage[utf8]{inputenc}
\usepackage[T1]{fontenc}
\usepackage{orcidlink}
\usepackage{verbatim}

\newcommand\mc{\mathcal}
\newcommand{\cont}[0]{\mathcal{C}}

\begin{document}
\title{
A Compositional Approach to
Diagnosing Faults in Cyber-Physical Systems
}
%
%
\author{Josefine B. Graebener\inst{1}\orcidlink{0000-0002-1376-0741} \and
Inigo Incer\inst{2}\orcidlink{0000-0001-7933-692X} \and
Richard M. Murray\inst{1}\orcidlink{0000-0002-5785-7481}}

%
%
\institute{University of Michigan, Ann Arbor, MI 48109 USA
\\ \email{iir@umich.edu}\\
}

\authorrunning{J.B. Graebener, I. Incer, and R.M. Murray}
%
\institute{California Institute of Technology, Pasadena CA 91125, USA 
\email{jgraeben@caltech.edu, murray@cds.caltech.edu}
\and
University of Michigan, Ann Arbor MI 48109, USA
\email{iir@umich.edu}\\}
\maketitle              
\begin{abstract}
Identifying the cause of a system-level failure in a cyber-physical system (CPS) can be like tracing a needle in a haystack.
This paper approaches the problem by assuming that the CPS has been designed compositionally
and that each component in the system is associated with an assume-guarantee contract.
We exploit recent advances in contract-based design that 
show how to compute the contract for the entire system using the component-level contracts.
When presented with a system-level failure, our approach is able to efficiently identify the components that are responsible for the system-level failure together with the specific predicates in those components' specifications
that are involved in the fault.
We implemented this approach using Pacti and demonstrate it through illustrative examples inspired by an autonomous vehicle in the DARPA urban challenge.

\keywords{assume-guarantee contracts \and diagnostics.}
\end{abstract}

\input{intro}

\input{background}
\input{tracing_guarantees}

\input{identifying_relevant_info}
\input{examples}
\input{conclusion}

%

\begin{credits}
\subsubsection{\ackname} This work was funded
by the Air Force Office of Scientific Research (grant number FA9550-22-1-0333).

\end{credits}
%
%
%
\bibliographystyle{splncs04}
\bibliography{main}
\end{document}

%% file: intro.tex
\newcommand{\revise}[1]{{\color{black} #1}}

\section{Introduction}
Any safety-critical system requires efficient detection of system abnormalities and faults to avoid dangerous situations and safety hazards.
The rising adoption of autonomy has only increased the need for such diagnostics~\cite{chen2012robust}. Faults are defined as deviations from the correct, expected system behavior, observable in at least one system property or parameter~\cite{van1997remarks}.
Diagnostics refers to the process of identifying, analyzing, and resolving problems that become apparent during the operation of a system, product, or process.
Diagnosing the cause of a system-level failure in a complex CPS can be a difficult process.

Fault diagnosis consists of three areas: detection, isolation, and identification~\cite{gao2015survey}. Detection refers to identifying when and where a fault occurs from the observable system output. Isolation considers the location of the fault, and identification refers to finding the type, shape, and size of the fault. There are four main fault diagnosis methods: model-based, signal-based, knowledge-based, hybrid, and active. In 1971, Beard introduced model-based fault diagnosis with the intent to replace hardware redundancy by analytical redundancy~\cite{beard1971failure}. Model-based fault diagnosis uses different techniques to monitor the actual system outputs and compare them to the predicted values.
Signal-based fault diagnosis uses measured signals instead of input-output models, and extracts features (or patterns) from a signal to make a diagnostic decision~\cite{gao2015survey}.
Knowledge-based fault diagnosis consists of a knowledge base and an inference engine~\cite{chi2022knowledge}. Hybrid approaches are a combination of the above-mentioned methods~\cite{gao2015survey}. Active fault detection is concerned with designing auxiliary input vectors to reveal faults~\cite{niemann2006setup}.

Diagnostics has been studied extensively in computer science and engineering. Some early and influential works include~\cite{de1987diagnosing} and ~\cite{reiter1987theory}. In~\cite{de1987diagnosing},
a model-based approach to diagnose faults in complex systems by observing symptoms and using reasoning techniques was introduced in 1987 by De Kleer and Williams. In 1987, Reiter developed a formal logical framework to diagnose faults consisting of three main components: a knowledge base, an observation base, and a set of inference rules~\cite{reiter1987theory}.
In formal methods, the problem of explaining why for certain robot specifications no implementing control strategy exists has been studied in~\cite{raman2012explaining}, while `repairing' specifications has been studied in~\cite{boteanu2017robot}. In~\cite{mallozzi2023contract}, assume-guarantee contract operators have been used for specification repair. Recently, TRACE, a tool for requirements analysis was introduced in~\cite{varanasi2025trace}, which uses SMT solvers to analyze system guarantees.

This paper presents a diagnostics approach that utilizes assume-guarantee reasoning and leverages the syntax of specifications to facilitate tracing the causes of violated system-level guarantees to potential subsystems.
The use of contracts in diagnostics simplifies the attribution of blame.
For example, suppose that we have a trace for a component. If this trace violates the assumption of the component, the component cannot be blamed for any undesired behavior. On the other hand, under satisfied assumptions, the behavior has to satisfy the promised guarantees. If the component does not deliver its guarantees in this case, then it did not satisfy its specification. Now we can further analyze the component and determine whether the implementation was faulty or if anything was missed when defining the specification.

\vspace{-2mm}
\paragraph{Problem definition.}
Suppose we implement a system with $n$ interconnected components having contracts $\cont_i = (\bigwedge_j a^i_j, \bigwedge_j g^i_j)$ for $i \in \{1, \ldots, n\}$, where $a^i_j$ and $g^i_j$ represent the $j$-th assumption and guarantee, respectively, of component $i
$.
Suppose that the specification of the entire system 
is given by
$\cont_S = (\bigwedge_j a_j, \bigwedge_j g_j)$ and that
we have access to a log file \texttt{Log} containing values for a subset of system variables such that the assumptions of $\cont_S$ are satisfied, but there is at least one guarantee
of $\cont_S$ that is violated.
{This paper introduces a technique
to identify both the contract $\cont_k$ and its specific assumptions and guarantees that were responsible for the system-level fault}. The technique can indicate exactly the predicates that need to be evaluated to diagnose the failure, as opposed to checking every single assumption and guarantee of all components in the system.
While the application of contracts in runtime verification and diagnostics has been pursued---see, e.g., ~\cite{10.1007/978-3-031-77382-2_5,10.1007/978-3-030-32079-9_10,10.1007/3-540-36577-X_24,10.1007/978-3-031-77382-2_25,10.1007/978-3-030-60508-7_1}---our work is the first to 
exploit the explicit computation of the contract operation of composition in order to determine not only what component in the system is the likely cause of the system-level fault, but also the specific predicates in the component's contract that are responsible for the fault.

\revise{
Our problem setup assumes that we have a log file witnessing a system-level failure and contracts for each component. We may ask, why should not we simply check every assumption and guarantee of all components? 
We believe that the targeted approach discussed in this paper is preferable for various reasons. First, 
the log file can contain several errors that are not related to the
system-level issue
we are interested in debugging. Lacking the means to pinpoint exactly the predicates that are involved in the system-level failure
can lead designers to long debugging campaigns of unrelated, but tempting, issues, resulting in distractions and loss of time, which may be costly in time-sensitive projects or when deadlines approach.
Second,
suppose there is a specific
system-level guarantee we want to monitor.
The methods discussed in this paper can locate exactly what predicates should be monitored in the system to ascribe blame to a component for the violation of the property we are tracking. Knowledge of these internal predicates can be used to instrument the system to monitor the desired predicates when running a second test, i.e., this methodology can be used in designing test campaigns.

}


\vspace{-2mm}
\paragraph{Contributions.}
We propose a diagnostics methodology based on contracts that enables a systematic search over system variables to trace violated system guarantees back to the responsible component.
Our approach reduces the number of predicates and components that need to be evaluated during fault localization by exploiting the explicit computation of the composition operation of contracts. We implement this methodology using Pacti~\cite{incer2025pacti}, a tool for compositional reasoning over assume-guarantee contracts. Finally, we demonstrate the effectiveness of our approach through illustrative examples and case studies inspired by autonomous vehicle behavior in the DARPA Urban Challenge~\cite{dutoitsensing}. A preliminary version of this work appeared in Chapter 5 of~\cite{graebener2024formal}.


%% file: background.tex
\section{Background}
The framework presented in this paper is based on contract-based design, first introduced as a design methodology for modular software systems~\cite{dijkstra1975guarded,lamport1990win,meyer1992applying} and later extended to cyber-physical systems~\cite{nuzzo2015platform,vincentelli2012taming}.
\begin{defi}[Assume-Guarantee Contract~\cite{Benveniste2008,incer2022thesis}]
Let $T$ be the language used to express specifications in our system. We assume this language has Boolean semantics.
A \emph{contract} is a pair $\cont = (a, g) \in T^2$, where $a$ are the assumptions, and $g$ the guarantees.
A model $E \models a$ is said to be an \emph{environment} of the contract $\cont$.
A model $M \models a \to g$ is said to be an \emph{implementation} of the contract $\cont$, meaning that $M$ provides the specified guarantees when it operates in an environment that satisfies the contract's assumptions.
There exists a preorder of contracts: we say $\cont_1$ is a refinement of $\cont_2$, denoted $\cont_1 \leq \cont_2$, if $(a_2 \leq a_1) \text{ and } (a_1 \to g_1 \leq a_2 \to g_2)$. 
Contracts $\cont_1$ and $\cont_2$ are said to be equivalent if $\cont_1 \le \cont_2$ and $\cont_2\le \cont_1$.
\end{defi}
As shown in~\cite{incer2024adjunction}, the contract algebra is a Stone algebra, but not a Boolean algebra.
A contract has three possible evaluations---not two---as requirements normally do (i.e., a requirement is either SAT or UNSAT).
A contract can evaluate to either \texttt{FAIL}, \texttt{ACTIVE}, or \texttt{IDLE}.
It evaluates to \texttt{FAIL} when the assumptions are satisfied but the guarantees are not. It evaluates to \texttt{ACTIVE} when both the assumptions and guarantees are satisfied. A contract evaluates to \texttt{IDLE} when the assumptions are not satisfied.
Given a system-level failure,
the diagnostics process corresponds to identifying the component contract that evaluates to \texttt{FAIL}.
It is key to our diagnostics process to compute the composition operation of contracts. A recent breakthrough in contract-based design showed how to compute this operation efficiently---see~\cite{incer2025pacti}.

The results of this paper are implemented in Pacti, an open-source Python package for compositional system analysis and design.
Components are defined using assume-guarantee contracts, and contract operations can be performed, such as composition, merging, and quotient. Contracts in Pacti are defined over a term algebra $T$ with Boolean semantics.
Pacti's basic data structure is the IO contract.
\begin{defi}[IO Contract~\cite{incer2025pacti}]
\label{def:iocontract}
Let $V$ be a set of variables. An \emph{IO contract} is the tuple $(I, O, \mathfrak{a}, \mathfrak{g})$, where $I,O \in V$ are disjoint sets of input and output variables respectively, $\mathfrak{a} \in T$ a set of assumptions, and $\mathfrak{g} \in T$ a set of guarantees. The terms in the assumptions only refer to input variables, and the terms in the guarantees only refer to input and output variables.
\end{defi}

The key contract operation for system-level design is composition, which yields the contract of a system obtained by interconnecting components represented using contracts.
The composition of contract $\mc{C}= (a,g)$ and contract $\mc{C'}= (a',g')$ can be directly computed as
\begin{equation}
    \label{kqjhfngbqjhkd}
    \mc{C}_c = \mc{C} \parallel \mc{C}' = ((a \land a') \lor (a \land \neg g) \lor (a' \land \neg g'), (g \lor \neg a) \land (g' \lor \neg a')).
\end{equation}
This yields the most refined contract that a system comprising two components, $M$ and $M'$, will satisfy, provided that $M$ and $M'$ were implemented such that they satisfy their corresponding contracts $\mc{C}$, and $\mc{C}'$, respectively. 
Whereas this operation satisfies optimality criteria,
the authors of~\cite{incer2025pacti} observe that the output of composition should be a contract expressed using only the variables that lie at the interface of the resulting system, which equation~\eqref{kqjhfngbqjhkd} does not provide.


To express contract composition using only interface variables,
equation~\eqref{kqjhfngbqjhkd} can be relaxed by either refining the assumptions, relaxing the guarantees, or both. 
For two assumptions $a$ and $a'$, we say that $a'$ refines $a$, denoted $a \geq a'$, if the denotation set corresponding to $a'$ is a subset of the set corresponding to $a$. 
Similarly, for guarantees $g$ and $g'$, we say that $g'$ is a relaxation of $g$, $g \leq g'$, if the denotation set corresponding to $g$ is a subset of the set corresponding to $g'$.
When thinking about assumptions and guarantees each as conjunctions of terms (or constraints), refining a contract informally corresponds to either assuming \emph{less}, guaranteeing \emph{more}, or both.
Pacti makes use of this contract relaxation to eliminate internal variables so that the returned composition is expressed only using the interface variables of the system.
Any relaxed contract that is computed in this way will be satisfied by a correct implementation of the system. Nevertheless, for a contract to be useful, we need to compute the relaxation systematically, as in the extreme case a contract that guarantees every possible behavior is a valid, yet pointless refinement in the context of capturing the system's behavior.

When computing a composition of contracts $\mc{C}$ and $\mc{C'}$, the assumptions of the composed contract are given as follows:
\begin{equation}
a_c = \overbrace{(a \land a')}^{\mathclap{\text{stem}}} \lor \overbrace{(a \land \neg g') \lor (a' \land \neg g)}^{\mathclap{\text{failure terms}}}.
\end{equation}

We refer to the first term as the \emph{stem}, as this is where the composed system should operate---where the assumptions of both components are satisfied. The second and third terms are referred to as \emph{failure terms}, where each term refers to one of the components having its assumptions satisfied but not delivering their guarantees.
As we want the composition to operate in the stem,
Pacti uses the failure terms to eliminate from the stem the variables that are not part of the interface of the resulting system.
These failure terms serve as the context for the elimination of variables in the stem---the details about this are contained in \cite{incer2025pacti}.
Once the failure terms are no longer needed to eliminate variables in the stem they are discarded, and we define the assumption of the relaxed contract as the transformed stem by refinement. The transformation of the guarantees follows a similar argument, but ensures that the variables are eliminated by relaxing the guarantees. The guarantees of a composition are given as follows:
\begin{align*}
g_c &= (g \lor \neg a) \land (g' \lor \neg a')
 = \underbrace{(g \land g')}_{\text{stem}} \lor (\neg a \land g') \lor (\neg a' \land g) \lor (\neg a \land \neg a'),
\end{align*}
where the \emph{stem} again refers to the desired area of operation, when both components satisfy their guarantees. The stem may contain variables that should be eliminated. We can use the remaining terms as a \emph{context} to transform the stem to eliminate these variables.
The transformation of either assumptions or guarantees is carried out in Pacti by functions that take as arguments the term to be transformed, the variables that need to be eliminated, and the context that can be used to carry out such elimination.

%% file: tracing_guarantees.tex
\section{Tracing System Guarantees}
This section 
presents a methodological approach to modeling the system and its constituent components to support an effective fault diagnosis mechanism. Our approach is based on the computation of contract operations in Pacti. We extended this tool to support diagnostics.

To begin, we assume that
each component in the system is modeled as an IO contract. 
For example, for a component $M$, we can define the corresponding IO contract $\mc{C} = (I, O, \mathfrak{a}, \mathfrak{g})$, where the set $\mathfrak{a}$ contains the term $i \le 2$ and the guarantee set $\mathfrak{g}$ contains the term $o \le 2i +1$, where $I = \{i\}$, and $O=\{o\}$ are the singleton sets of the input and output variables.

\begin{defi}[Faulty component]
Given a component $M$ and the corresponding contract $\mc{C}=(I,O,\mathfrak{a}, \mathfrak{g})$, $M$ is \emph{faulty} if it contains a behavior that does not satisfy the guarantees $\mathfrak{g}$, but the assumptions $\mathfrak{a}$ are satisfied. 
\end{defi}

As discussed in the background section, Pacti uses a filtering procedure to determine the relevant context terms when computing the assumptions and guarantees of the composed contract. These context terms come from the assumptions and guarantees of the contracts being composed.
We extended Pacti with an ID system to keep track of which context terms were used to generate a resulting term for the assumptions and guarantees of the composed contract. For each composition operation, we can thus define a composition graph that allows us to map the composed assumption and guarantee terms to the terms that were used in their transformation.
A composition graph consists of a set of vertices, where each vertex corresponds to a term in the assumptions or guarantees of the individual contracts, and the composed contract.
The edges in the composition graph connect vertices if the corresponding individual contract term was used to generate the composed contract term, shown as the edges from left to right in Figure~\ref{fig:graph_single_composition}.

\begin{defi}[Composition Graph]
Let components \( M_1 \) and \( M_2 \) have corresponding IO contracts 
\( \cont_1 = (I_1, O_1, \mathfrak{a}_1, \mathfrak{g}_1) \) and 
\( \cont_2 = (I_2, O_2, \mathfrak{a}_2, \mathfrak{g}_2) \), 
with their composition given by 
\( \cont= (I, O, \mathfrak{a}, \mathfrak{g}) \). 
A \emph{composition graph} is a directed graph 
\( G = (V, E) \), where each vertex in \( V \) corresponds to a term 
from the assumptions or guarantees of 
\( \cont_1 \), \( \cont_2 \), or \( \cont \). 
Specifically, \( V_{i,a} \subseteq V \) and \( V_{i,g} \subseteq V \) 
represent the assumptions \( \mathfrak{a}_i \) and guarantees 
\( \mathfrak{g}_i \) of component \( i \), while 
\( V_a \subseteq V \) and \( V_g \subseteq V \) represent the composed 
assumptions and guarantees. 
For simplicity, we use the same symbol \( s \) to refer to both a term 
and its corresponding vertex \( s \in V \). 
An edge \( (u, v) \in E \) exists if the term \( u \) was used to generate 
term \( v \). 
A path from vertex \( s \) to \( t \) in \( G \), denoted 
\( \texttt{path}(G, s, t) \), exists if there is a sequence of edges 
connecting the vertices.
\end{defi}

\begin{example}
\label{ex:single_composition}
Given two components $M_1$ and $M_2$ and their inputs and outputs as illustrated in Figure~\ref{fig:ex_single_composition} and their IO contracts as $\mc{C}_1 = (\{i\}, \{o\}, \mathfrak{a}_1, \mathfrak{g}_1)$, where
 $\mathfrak{a}_1 = \{ i \geq 0, i \leq 2 \} \text{, and } \mathfrak{g}_1 =  \{ o + i \leq 3\}$ and $\mc{C}_2 = (\{o\}, \{o'\}, \mathfrak{a}_2, \mathfrak{g}_2)$, where 
 $\mathfrak{a}_1 = \{ o \leq 5 \} \text{, and } \mathfrak{g}_1 =  \{ o + 2o' \geq 6\}$. The composition results in the contract $\mc{C} = (\{i\}, \{o'\}, \mathfrak{a}, \mathfrak{g})$, where $\mathfrak{a} = \{ i \geq 0, i \leq 2 \}$, and $\mathfrak{g} =  \{ i - 2o' \leq -3\}$.
The composition graph corresponding to the composition $\mc{C}_1 \parallel \mc{C}_2$ is shown in Figure~\ref{fig:graph_single_composition}.
\end{example}

From our experience,
the transformation of an assumption or guarantee usually involves few context terms. This means that resulting composition graph is very sparse.
This allows us to trace a composition-level guarantee back to a small set of terms on the component level.

All operations in Pacti are computed under the assumption that the components operate correctly as specified by their contracts. The entire premise of testing lies in the difference between a `perfectly' specified system and its real-world implementation.
In a real-world system, specifications might be incomplete, or implementations might be faulty. A single component failure might present itself in multiple ways; it might result in a system-level guarantee violation, or it might not. If a component failure is latent, meaning it does not show itself in the system-level guarantee violation, our proposed framework cannot detect it. If the fault does show itself in a system-level guarantee violation, we can trace it to the responsible component(s) by tracking which relevant component's guarantees were violated under satisfied assumptions.
\vspace{-5mm}

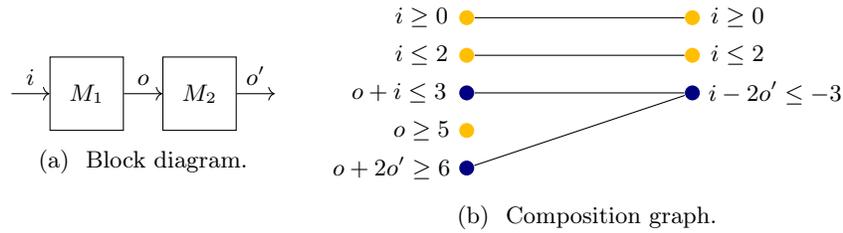
\begin{figure*}
\centering
\begin{minipage}{0.3\textwidth}
\centering
\tikzstyle{block} = [draw, fill=white, rectangle, 
    minimum height=3em, minimum width=3em]
\tikzstyle{init} = [pin edge={to-,thin,black}]
\begin{tikzpicture}[node distance=2.5cm,auto]
    \node [block] (a) {$M_1$};
    \node (init) [left of=a,node distance=1cm, coordinate] {};
    \node [block] (c) [right of=a, node distance=1.5cm]{$M_2$};
    \node (topend) [right of=c,node distance=1cm, coordinate] {};
    \path[->] (init) edge node {$i$} (a);
    \path[->] (a) edge node {$o$} (c);
    \path[->] (c) edge node {$o'$} (topend);
\end{tikzpicture}
\subcaption{\label{fig:ex_single_composition} Block diagram.}
\end{minipage}
\begin{minipage}{0.65\textwidth}
\centering
\begin{tikzpicture}
\definecolor{dandelion}{rgb}{1.0, 0.75, 0.0}
\tikzstyle{orange}=[circle, fill=dandelion, inner sep=2pt]
\tikzstyle{blue}=[circle, fill=blue!50!black, inner sep=2pt]

\node[orange] (o1) at (0,1.5) {}; \node[left of=o1, node distance=6mm]{\small $i \geq 0$};
\node[orange] (o2) at (0,1) {}; \node[left of=o2, node distance=6mm]{\small $i \leq 2$};
\node[blue]   (b1) at (0,0.5) {}; \node[left of=b1, node distance=9mm]{\small $o + i \leq 3$};
\node[orange] (o3) at (0,0) {}; \node[left of=o3, node distance=6mm]{\small $o \geq 5$};
\node[blue]   (b2) at (0,-0.5) {}; \node[left of=b2, node distance=10mm]{\small $o + 2o' \geq 6$};

\node[orange] (o1r) at (3,1.5) {}; \node[right of=o1r, node distance=6mm]{\small $i \geq 0$}; 
\node[orange] (o2r) at (3,1) {}; \node[right of=o2r, node distance=6mm]{\small $i \leq 2$};
\node[blue]   (b1r) at (3,0.5) {}; \node[right of=b1r, node distance=11mm]{\small $i - 2o' \leq -3$};

\draw (o1) -- (o1r);
\draw (o2) -- (o2r);
\draw (b1) -- (b1r);
\draw (b2) -- (b1r);

\end{tikzpicture}
\subcaption{\label{fig:graph_single_composition} Composition graph.}
\end{minipage}

\caption{Block diagram and composition graph for Example~\ref{ex:single_composition}.}
\label{fig:single_composition}
\end{figure*}
\vspace{-5mm}

At this point, we can create a composition graph for the composition of two components and their corresponding IO contracts. To build the overall system, we need to compose multiple components. Contract composition is a binary operation. Therefore, to compose the entire system, we need to compose two components at a time, and then compose their composition with the next component.
There are many ways of composing the same system, and the resulting contract for the composition is dependent on the order of composition. As Pacti hides internal variables, the composition has to be chosen carefully such that the variables necessary for future compositions are kept. 
Another important aspect that can guide the composition order is the availability of component data. When there is a lack of available information from inside a block of components it might be beneficial to compose these components first and treat them as a meta-component---a grouping of multiple components. If the analysis ends up pointing to this meta-component as the possible cause, a more detailed analysis can still be set up focusing on these components.
For this framework, we assume that a composition order has been chosen. This order will be maintained for the remainder of the diagnostics process.

\begin{defi}[Composition Order]
\label{def:composition_order}
The component contracts are given as a \emph{composition order} 
\( \mathtt{CompOrd} = [\cont_1, \ldots, \cont_N] \), 
where \( N \) is the number of components. 
\( \mathtt{CompOrd} \) specifies the order in which contracts are composed 
to compute the system specification. 
The composition up to the \( k^{th} \) contract is denoted 
\( \cont_{\text{comp},k} \coloneqq \cont_1 \parallel \ldots \parallel \cont_k \).
\end{defi}
This composition order requires composing the system starting from a single component and building the system up from there. If it is desired to start by composing certain regions of the system first, this framework can easily be extended to include this approach. Otherwise, the component contracts in the composition order need to be provided at the level of granularity such that they can be composed according to a composition order defined in Definition~\ref{def:composition_order}.
We will now define the diagnostics graph that corresponds to the composition of multiple components and outline how it is constructed. The union of two graphs $G_1$ and $G_2$ is defined as $G = (V_1 \cup V_2, E_1 \cup E_2)$, and we will denote it by $G = G_1 \cup G_2$. Simply stated, the diagnostics graph consists of multiple composition graphs. The system is composed step-by-step according to the composition order, and for every composition, the diagnostics graph is extended by the composition graph for this composition. An example is shown in Figure~\ref{fig:three_components_ex_graph}, where we can see the union of two composition graphs. Each `column' of vertices corresponds to the individual contract terms in a composition, where the top nodes correspond to the already composed system (or the first component for the first composition), and the bottom vertices correspond to the next contract in the composition order.

\begin{defi}[Diagnostics Graph]
\label{def:diagnostics_graph}
Given component contracts in a composition order $\mathtt{CompOrd} = [\mc{C}_1, \ldots, \mc{C}_N]$, the  \emph{diagnostics graph} $G=(V,E)$ is constructed as follows.
For each $i$, $2 \leq i \leq N$, we compute the composition $\mc{C}_{\text{comp,i-1}} \parallel \mc{C}_i$ and the corresponding graph $G_i$. Then the diagnostics graph $G$ is defined as $G=\bigcup_{i=2}^N G_i$.
\end{defi}

\begin{defi}[Diagnostics Map]
\label{def:causality_map}
Let $\mathtt{CompOrd}$ be the composition order for $N$ components and their contracts. Let $\mc{C} = \parallel_{i=1}^N\mc{C}_i = (\mathfrak{a}, \mathfrak{g})$ be the system-level composition according to the composition order,
and let the corresponding diagnostics graph be $G$. Then, we can define the \emph{diagnostics map} $\mathtt{CM}: \mathfrak{g} \rightarrow 2^{(\bigcup_{i=1}^N(\mathfrak{a}_i \cup \mathfrak{g}_i)) \times \{\mc{C}_i \}_{i=1}^N}$, that maps each composed assumption and guarantee term to a set of component level assumption or guarantee terms through the diagnostics graph. That is, for system-level term $s$ and the corresponding vertex $s \in V$, we have 
\begin{equation}
\label{eq:causality_map}
\begin{split}
\mathtt{CM}(s) =& \{ (t, \mc{C}_i) \:\vert \forall \mc{C}_i \in \mathtt{CompOrd},\: t\: \in  \mathfrak{g}_i \cup \mathfrak{a}_i, \text{ if } \\
&\exists \mathtt{path}(G,t,s) \text{ and } \forall u\in V, u \neq t \implies \:\nexists (u,t) \in E\},
\end{split}
\end{equation}
where $t \in V$ corresponds a component-level term used to generate $s$, and $i$ is the index of the component the term belongs to.
\end{defi}
The diagnostics map finds the leaf nodes in the diagnostics graph that have a path to the vertex corresponding to the violated guarantee. For each leaf node, it returns the term corresponding to the leaf node and the contract that this term belongs to in the form of a tuple. With this information, we can now focus our attention on these terms first and start the diagnostics process by evaluating these terms using the available test data.

%% file: identifying_relevant_info.tex
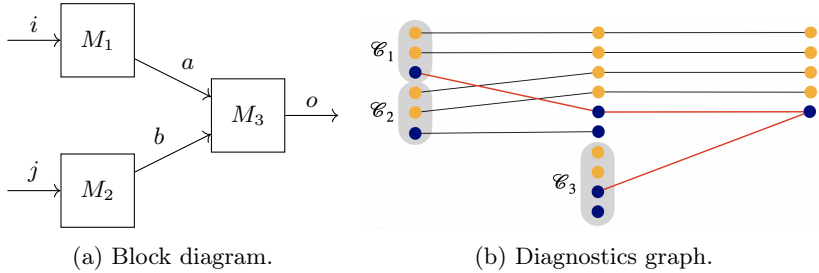
\begin{figure}[t]
    \centering
    \begin{minipage}[b]{0.4\linewidth}
        \centering
        \tikzstyle{block} = [draw, fill=white, rectangle, 
            minimum height=3em, minimum width=3em]
        \tikzstyle{init} = [pin edge={to-,thin,black}]
        \begin{tikzpicture}[node distance=2cm,auto]
            \node [block] (c1) {$M_1$};
            \node (init) [left of=c1,node distance=1.2cm, coordinate] {};
            \node (center) [below of=c1,node distance=1cm, coordinate] {};
            \node [block] (c2) [below of=c1, node distance=2cm]{$M_2$};
            \node [block] (c3) [right of=center, node distance=2cm]{$M_3$};
            \node (topend) [right of=c3,node distance=1.2cm, coordinate] {};
            \node (bottominit) [left of=c2,node distance=1.2cm, coordinate] {};
            \path[->] (init) edge node {$i$} (c1);
            \path[->] (bottominit) edge node {$j$} (c2);
            \path[->] (c1) edge node {$a$} (c3);
            \path[->] (c2) edge node {$b$} (c3);
            \path[->] (c3) edge node {$o$} (topend);
        \end{tikzpicture}
        \subcaption{Block diagram.}
        \label{fig:three_components_ex}
    \end{minipage}
    \begin{minipage}[b]{0.5\linewidth}
        \centering
        \includegraphics[width=\linewidth]{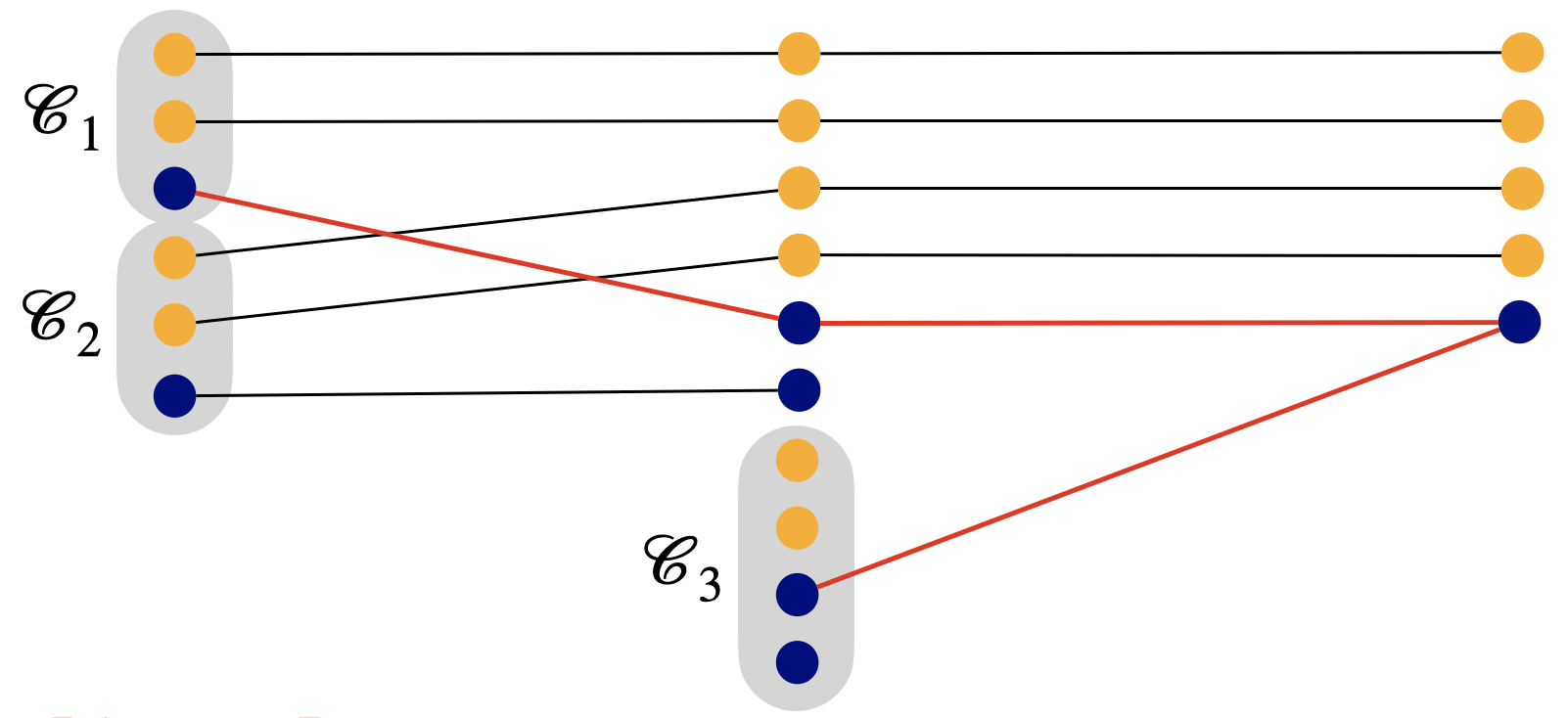}
        \subcaption{Diagnostics graph.}
        \label{fig:three_components_ex_graph}
    \end{minipage}
    \label{fig:assumption_trace_ex}
    \caption{Block diagram and diagnostics graph for composition in Example~\ref{ex:three_components_ex}.}
    \vspace{-5mm}
\end{figure}


Using the diagnostics graph and the system block diagram, we can identify the component assumptions and guarantees that were relevant to the generation of the top-level system guarantee that was violated. We can then focus on the relevant components of the system block diagram and analyze whether the guarantee was satisfied or violated. Once we find a component where the guarantee was violated we can shift our focus to the assumptions of this term. From then on two different scenarios can occur: i) the assumptions are satisfied, or ii) the assumptions are violated. \\
\textbf{Case i)} When a contract's assumptions are satisfied, but the component does not deliver the contract's guarantees, it either means that the component failed or
that the contract did not adequately characterize the context of operation of the component. 
The analysis at this level terminates, and the component designer needs to analyze the behavior of this component.
\\
\textbf{Case ii)} In the case of violated assumptions, the component is likely not the cause of the failure, as another component's behavior resulted in the violation of the assumptions. To diagnose this fault, we need to trace the violated assumption back to the guarantees of one or more components.
We do this by searching over
the composition order for
the instance of a composition of the
component whose assumptions were violated.
As this component's contract was composed with another contract,
we identify the guarantees of this second contract that are relevant for the satisfaction of the violated assumption.
The function \textproc{ElimVarsbyRefinement} in Pacti
identifies for us the a subset of terms that participate in the satisfaction of formulas---see \cite{incer2025pacti} for details about this function.
After we identify the assumptions and guarantees of the second contract,
we refer to the diagnostics graph again to trace these newly identified guarantees back.

It is important to note that this process of tracking assumptions only applies to components whose assumptions are not solely dependent on the overall system input variables. 
A system-level failure is defined as having satisfied the system-level assumptions, but failing to satisfy the guarantees of the system. Thus, the system-level assumptions are satisfied---this will ensure that component-level assumptions that are only dependent on the system-level input variables are also satisfied.

\paragraph{Identifying Causes for Violated Assumptions.}
Assume we are given the component $M$, its corresponding contract $\mc{C} = (I,O,\mathfrak{a}, \mathfrak{g})$, the component $M_{\text{other}}$ in the composition order that is composed with component $M$, and the contract $\mc{C}_{\text{other}}= (I_{\text{other}},O_{\text{other}},\mathfrak{a}_{\text{other}}, \mathfrak{g}_{\text{other}})$
corresponding to $M_{\text{other}}$.
The component assumption that was violated is denoted $a_v \in \mathfrak{a}$. We will use Pacti to find the relevant context used to refine this assumption, referred to as the function $\textproc{FindCauseForAssumption}(a_v, \mathfrak{a}_{\text{other}} \cup \mathfrak{g}_{\text{other}})$.
This function exploits the function call \textproc{ElimVarsbyRefinement} in Pacti, which
transforms the assumption $a_v$ with the use of $\mathfrak{a}_{\text{other}} \cup \mathfrak{g}_{\text{other}}$ as the context to eliminate any unwanted variables.
We can make use of the same function augmentation that we created to compute the diagnostics graph to analyze the transformation at this level.
The instrumentation of the filtering step will return the relevant context terms $\mathfrak{c}_r \subseteq \mathfrak{a}_{\text{other}} \cup \mathfrak{g}_{\text{other}}$ in the assumptions and guarantees of $\mc{C}_{\text{other}}$. Once we have determined $\mathfrak{c}_r$, the diagnostics map allows us to trace back the terms in $\mathtt{CM}(c)$ for each $c \in \mathfrak{c}_r$ to the responsible component level terms.
The entire diagnostics procedure is outlined in Algorithm~\ref{alg:diagnostics}.

\input{diagnostics_alg}

This approach can identify multiple component faults under certain conditions. As discussed above, we can only find the faulty component if a system-level guarantee is violated. If two component faults end up cancelling each other out (i.e., are not observable at the system level), then this approach cannot identify them as no system-level guarantee was violated. If a faulty component results in violated assumptions for another component, we cannot determine whether the component with the violated assumption also failed. This is due to the fact that for a contract with violated assumptions, any behavior is allowed. Under the condition that all faulty components are independent (i.e., a faulty component does not lead to violated assumptions of another faulty component), this procedure is able to identify all faulty components.



\begin{example} 
\label{ex:three_components_ex}
Let there be a system consisting of three component contracts in the composition order  $\mathtt{CompOrd} = [\mc{C}_1,\mc{C}_2, \mc{C}_3]$ with their inputs and outputs as illustrated in Figure~\ref{fig:three_components_ex}. The IO contracts are given as $\mc{C}_1 = (\{i\}, \{a\}, \mathfrak{a}_1, \mathfrak{g}_1)$, where
 $\mathfrak{a}_1 = \{i \leq 2, i \geq 0\} \text{, and } \mathfrak{g}_1 =  \{a \leq 2\}$ and $\mc{C}_2 = (\{j\}, \{b\}, \mathfrak{a}_2, \mathfrak{g}_2)$, where $\mathfrak{a}_2 = \{j \leq 2, j \geq 0\} \text{, and } \mathfrak{g}_2 =  \{b \leq 3\}$ and $\mc{C}_3 = (\{a,b\}, \{o\}, \mathfrak{a}_3, \mathfrak{g}_3)$, where $\mathfrak{a}_3 = \{a \leq 5, b \leq 5\} \text{, and } \mathfrak{g}_3 =  \{o \leq a, o \leq b\}$.
 The system-level contract is computed as $\mc{C} = (\{i,j \}, \{ o\}, \mathfrak{a}, \mathfrak{g})$, where $\mathfrak{a} = \{i \leq 2, i \geq 0, j \leq 2, j \geq 0 \}$, and $\mathfrak{g} = \{ o \leq 2\}$.
 Suppose the following trace is observed during execution: $i = 1$, $j = 1$, $a = 2$, $b= 7$, $o=3$.
 The system-level guarantee $g_v \coloneqq o \leq 2$ is violated while the assumptions are satisfied.
The diagnostics map is given by $\mathtt{CM}(g_v) = \{ (a \leq 2, \mc{C}_1), (o \leq a, \mc{C}_3) \}$.
Applying Algorithm~\ref{alg:diagnostics},
we check the guarantee of component $M_1$ first and see that $a \leq 2$ is satisfied. Next, we check $o \leq a$, which is not satisfied, as $3 \not\leq 2$. This narrows our analysis on component $M_3$.
Evaluating the assumptions of contract $\mc{C}_3$, we find that $a \leq 5$ is satisfied, but $b \leq 5$ is not. Therefore, $M_3$ is not responsible for the violation. We can now trace which terms were used to transform $b \leq 5$ to find which terms to evaluate next in our search for the failed component. For this, we compute $\mc{C}_{\text{other}} = \mc{C}_1 \parallel \mc{C}_2$ from the composition order and evaluate 
 $\textproc{FindCauseForAssumption}(b \leq 5, \mathfrak{a}_{\text{other}} \cup \mathfrak{g}_{\text{other}})$, which returns the relevant context term as the following guarantee from component contract $\mc{C}_2$, $b \leq 3$. This guarantee is not satisfied, as $7 \not\leq 3$.
 Subsequently, we check the assumptions for $\mc{C}_2$, $0 \leq j \leq 2$, which are satisfied, leading to the identification of $M_2$ as the component responsible for the violation.
In this example, only 6 terms were checked instead of all 10, showing that tracing the cause of a violated assumption can require fewer checks than evaluating all component-level terms from the log.
\end{example}

\begin{theorem}
Suppose we have a list of components $M_1,\ldots,M_N$, their contracts in a composition order $\mathtt{CompOrd} = [\mc{C}_1,\ldots \mc{C}_N]$, a violated system-level guarantee $g_v$, and the complete log data of a failing trace $\mathtt{Log}$. If $g_v$ is a guarantee of the composed system, we can identify the faulty component(s) using Algorithm~\ref{alg:diagnostics}.
\end{theorem}
\begin{proof}
Let us denote the composed contract as $\mc{C} = (I,O,\mathfrak{a}, \mathfrak{g})$. For a given composition and the corresponding contract $\mc{C}$, under satisfied system-level assumptions $\mathfrak{a}$ and a violated system-level guarantee $g_v \in \mathfrak{g}$, by construction of the composition, there exists at least one faulty component.
From $\mathtt{CompOrd}$, we can construct a diagnostics map $\mathtt{CM}$. If $g_v \in \mathfrak{g}$, $\mathtt{CM}(g_v)$ is guaranteed to contain at least one component-level guarantee $g_k \in \mathfrak{g}_k$, a guarantee of contract $\mc{C}_k$, where $1 \leq k \leq  N$. For each $g \in \mathtt{CM}(g_v)$, we evaluate from the trace $\mathtt{Log}$ whether it is satisfied or violated. If $g_k$ is violated, we have two different cases: i) if the assumptions $\mathfrak{a}_k$ of contract $\mc{C}_k$, are satisfied, then $M_k$ is added to the list of responsible components; in case ii), if the assumptions of $\mc{C}_k$ are violated, from the composition operation in Pacti, we can identify which component-level terms were used to refine this assumption and identify which terms to evaluate next. 
By definition of assume-guarantee contracts, an implementation of a contract where the assumptions are satisfied and whose guarantees are violated is faulty. Any component in the analysis that violated its guarantees under satisfied assumptions is faulty.
\end{proof}

%% file: diagnostics_alg.tex
\floatname{algorithm}{Algorithm}
\renewcommand{\algorithmicrequire}{\hspace*{\algorithmicindent}\quad\textbf{Input:}}
\renewcommand{\algorithmicensure}{\hspace*{\algorithmicindent}\quad\textbf{Output:}}

\begin{algorithm}
\caption{Diagnosing Violated Guarantee $g_v$}\label{alg:diagnostics} 
\begin{algorithmic}[1] 

\Procedure{Diagnose}{$g_v$, $\mathtt{CompOrd}$, $\mathtt{Log}$}
\Require failed guarantee $g_v$, composition order $\mathtt{CompOrd} = [\mc{C}_1, \ldots, \mc{C}_N]$, log data $\mathtt{Log}$
\Ensure set of failed components $C_f$
\State G $\leftarrow (\emptyset,\emptyset)$ \Comment{Initialize empty diagnostics graph}
\For{$\mc{C}_i \in \mathtt{CompOrd}$}
\State $\mc{C}_{\text{comp},i-1} \leftarrow \mc{C}_1\parallel \ldots \parallel\mc{C}_{i-1}$
\State $G_i \leftarrow$ \textproc{CompositionGraph}($\mc{C}_{\text{comp},i-1} ,\mc{C}_i)$
\State $G \leftarrow G \cup G_i$ \Comment{Add composition graph to diagnostics graph}
\EndFor
\State  $\mathtt{CM} \leftarrow $ define diagnostics map according to equation~\eqref{eq:causality_map}
\State $C_f \leftarrow \textproc{Trace}(g_v, \mathtt{CompOrd}, \mathtt{CM})$ \Comment{Find set of failed components}
\State \textbf{return} $C_f$
\EndProcedure
\State

\Procedure{Trace}{$g_v$, $\mathtt{CompOrd}$, $\mathtt{CM}$, $\mathtt{Log}$}
\Require guarantee to trace $g_v$, composition order $\mathtt{CompOrd} = [\mc{C}_1, \ldots, \mc{C}_N]$, diagnostics map $\mathtt{CM}$, log data $\mathtt{Log}$, components $M_1,\ldots, M_N$
\Ensure set of failed components $C_f$
\State $C_f \leftarrow \emptyset$ \Comment{Initialize empty set of failed components}
\For{$(t,\mc{C}_i) \in \mathtt{CM}(g_v)$} \Comment{Component-level term $t$, component index $i$}
\If{$t \in \mathfrak{g}_i$} \Comment{$\mathfrak{g}_i$ are the guarantees of $\mc{C}_i$}
\If{$\textproc{NotSatisfied}(t, \mathtt{Log})$} \Comment{Check if $t$ is satisfied in log data}
\State $\mathtt{AssumptionsSatisfied} \leftarrow \text{True}$ \Comment{Initialize flag as True}
\For{$a_i \in \mathfrak{a}_i$}
\If{$\textproc{NotSatisfied}(a_i, \mathtt{Log})$}\Comment{Check $a_i$ in log data}
\State $\mathtt{AssumptionsSatisfied} \leftarrow \text{False}$
\State $\mc{C}_{\text{other}} \leftarrow \mc{C}_1\parallel \ldots \parallel\mc{C}_{i-1}$
\State $\mathfrak{c} \leftarrow $\textproc{FindCauseForAssumption}($a_i$, $\mathfrak{a}_\text{other} \cup \mathfrak{g}_\text{other}$)
\For{$c_k \in \mathfrak{c}$}
\State $C_f \leftarrow C_f \cup \textproc{Trace}(c_k, \mathtt{CompOrd}, \mathtt{CM})$
\EndFor
\EndIf
\EndFor
\If{$\mathtt{AssumptionsSatisfied}$} 
\State $C_f \leftarrow C_f \cup M_i$ \Comment{Add component $i$ to the list}
\EndIf
\EndIf
\EndIf
\EndFor
\State \textbf{return} $C_f$
\EndProcedure
\State

\Procedure{FindCauseForAssumption}{$a_i$, $\mathfrak{a}_\text{other} \cup \mathfrak{g}_\text{other}$}
\Require assumption to trace $a_i$, context terms $\mathfrak{a}_\text{other} \cup \mathfrak{g}_\text{other}$
\Ensure list of relevant context terms $\mathfrak{c}$
\State $\_,\mathfrak{c} \leftarrow \textproc{ElimVarsbyRefinement}(a_i, \mathfrak{a}_\text{other} \cup \mathfrak{g}_\text{other})$ \Comment{Refine the assumption $a_i$ using the context terms $\mathfrak{a}_\text{other} \cup \mathfrak{g}_\text{other}$ and return the relevant context terms}
\State \textbf{return} $\mathfrak{c}$
\State 

\EndProcedure
\end{algorithmic}
\end{algorithm}

%% file: examples.tex
\section{Examples}
The following example was inspired by Alice, Caltech's entry in the 2007 DARPA Urban Challenge. While conducting the pre-challenge testing campaign, Alice faced scenarios where it failed to accomplish its objective because of an unforeseen behavior arising from the interaction of various subsystems in that particular situation. 
In this section, we will illustrate a scenario that is loosely based on the real-world scenarios that Alice faced in the DARPA Urban Challenge, which are described in detail in~\cite{dutoitsensing}.

\smallskip
\noindent
\textbf{Alice at Intersection using Propositional Logic}
In this example, the test was set up such that Alice was approaching an intersection with multiple cars already waiting at the intersection as shown in Figure~\ref{fig:alice_4_way}. While Alice was approaching, its sensors detected the other cars in the intersection and commanded Alice to stop and give way to the other cars. The unforeseen circumstance was that the deceleration tilted the LIDARs forward and towards the ground such that Alice lost sight of the other cars momentarily. Once Alice came to a full stop, the line of sight of the LIDARs tilted back up and detected the cars again, but now Alice was under the impression that the cars just arrived, leading to the control system commanding Alice to drive into the intersection and resulting in unsafe behavior.

We model the components in Alice's control architecture as shown in Figure~\ref{fig:alice_intersection_diagram}.
Alice's system consists of three components, the \emph{Perception}, the \emph{Planner} and the \emph{Tracker} (highlighted in the red dashed box). 
The real world is captured in the component labeled \emph{World}, that provides the real-world observations. The \emph{Safety} component evaluates Alice's behavior and returns a safety flag that evaluates to True if Alice's behavior is safe and False otherwise.
For each component in Alice's system, we can define an IO contract that describes the correct component behavior.
The perception component is modeled as follows
\begin{equation*}
\mc{C}_{\text{perception}}= \{I_P, O_P, \mathfrak{a}_{\text{perception}}, \mathfrak{g}_{\text{perception}} \},
\end{equation*}
with the input variables $I_{\text{perception}} = \{ c^i_{T_1}, c^i_{T_2},c^i_{T_3}, \text{poor\_visibility}\}$, the output variables are 
$O_{\text{perception}}=\{c^i_{P_1}, c^i_{P_2},c^i_{P_3}\}$, the assumptions 
$\mathfrak{a}_{\text{perception}} = \{\neg\text{poor\_visibility}\}$, and guarantees
$\mathfrak{g}_{\text{perception}} = \{c^i_{T_1} \Leftrightarrow c^i_{P_1}, c^i_{T_2} \Leftrightarrow c^i_{P_2}, c^i_{T_3} \Leftrightarrow c^i_{P_3}\}$.
The variables $c^i_{T_1}, c^i_{T_2},c^i_{T_3}$ correspond to whether there is a car in the $1^{st}$, $2^{nd}$, and $3^{rd}$ position in the intersection, and $c^i_{P_1}, c^i_{P_2},c^i_{P_3}$ corresponds to the perceived state of the cars in the intersection; the variable $\text{poor\_visibility}$ represents to the visibility conditions.
This contract describes that the perception component guarantees that the cars in the intersection will be detected correctly if there is no poor visibility.
\vspace{-5mm}
\usetikzlibrary{arrows.meta, positioning}
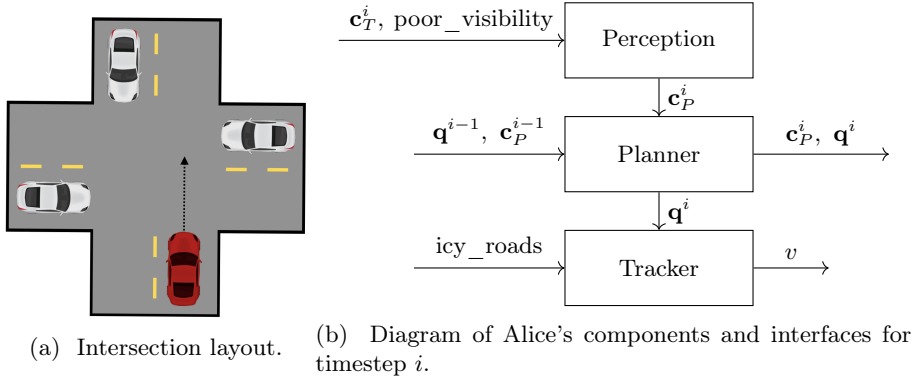
\begin{figure*}
\centering
\begin{minipage}{.33\textwidth}
\hspace{1mm}
    \includegraphics[width=\linewidth,trim={0.0cm 0.0cm -0.0cm 0.0cm}]{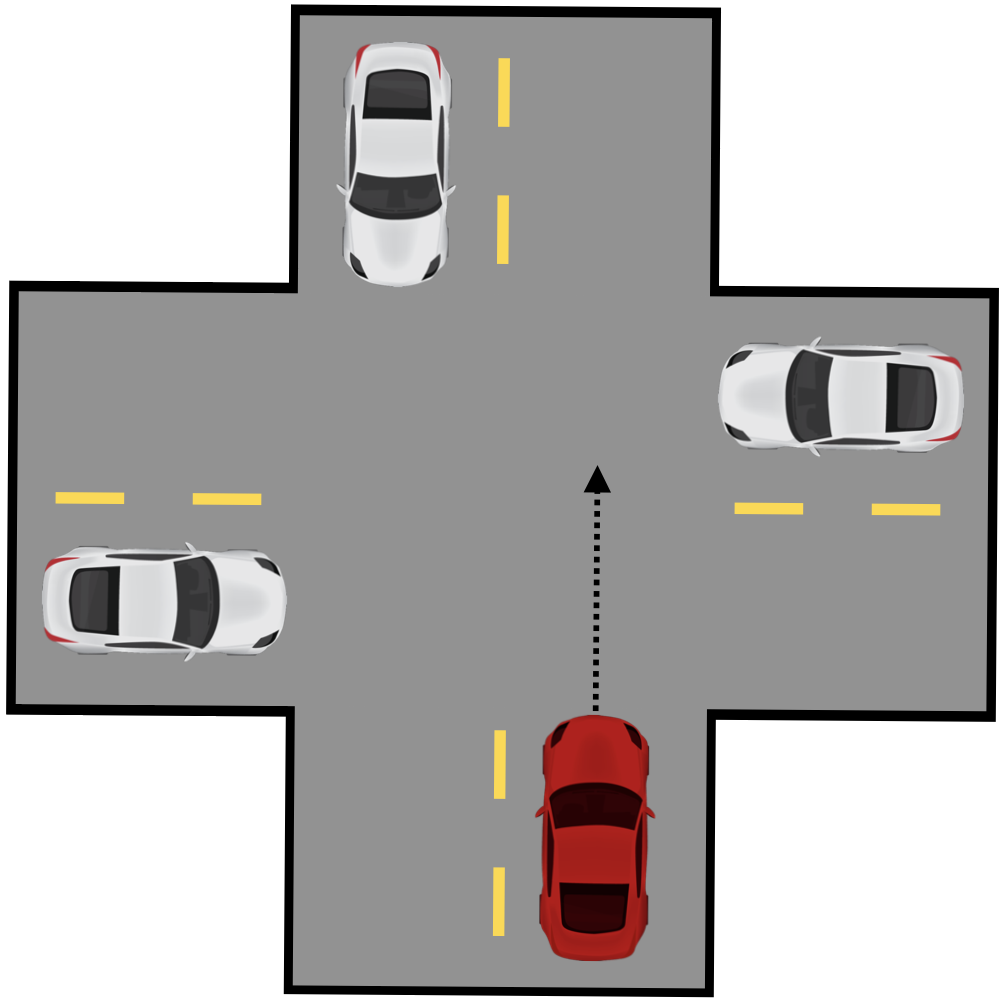}
\subcaption{\label{fig:alice_4_way} Intersection layout.}
  \end{minipage}
  \begin{minipage}{.65\textwidth}
    \centering
\begin{tikzpicture}[->, node distance=0.5cm and 1cm]

\node[draw, minimum width=2.5cm, minimum height=1cm] (perception) {Perception};
\node[draw, below=of perception, minimum width=2.5cm, minimum height=1cm] (planner) {Planner};
\node[draw, below=of planner, minimum width=2.5cm, minimum height=1cm] (tracker) {Tracker};

\draw (perception) -- node[right] {$\mathbf{c}_P^i$} (planner);
\draw (planner) -- node[right] {$\mathbf{q}^i$} (tracker);

\draw ([xshift=-3.0cm]perception.west) -- node[above] {$\mathbf{c}_T^i,\:\text{poor\_visibility}$} (perception.west);
\draw ([xshift=-2cm]planner.west) -- node[above] {$\mathbf{q}^{i-1},\ \mathbf{c}_P^{i-1}$} (planner.west);
\draw ([xshift=-2cm]tracker.west) -- node[above] {icy\_roads} (tracker.west);
\draw (planner.east) -- ++(1.8,0) node[midway, above] {$\mathbf{c}_P^i,\ \mathbf{q}^i$};
\draw (tracker.east) -- ++(1.0,0) node[midway, above] {$v$};

\end{tikzpicture}
    
\subcaption{\label{fig:alice_intersection_diagram} Diagram of Alice's components and interfaces for timestep $i$.}
  \end{minipage}
\vspace{-2mm}
\caption{Layout of the intersection and Alice's component block diagram.}
\label{fig:alice_intersection_failure}
\vspace{-5mm}
\end{figure*}

\noindent
The planner component is tasked with determining Alice's spot in the queue of arriving cars to determine whether Alice has the right of way or needs to stop for other cars to take their turn first.
For this, the Planner needs to keep track of the arrival order of the cars at the intersection.
The planner inputs are 
$$I_{\text{planner}} = \{c_{P_1}^i, c_{P_2}^i, c_{P_3}^i, c_{P_1}^{i-1}, c_{P_2}^{i-1}, c_{P_3}^{i-1},  q^{i-1}_1, q^{i-1}_2, q^{i-1}_3, q^{i-1}_4\},$$
where $c_{P_1}^i, c_{P_2}^i, c_{P_3}^i$ are the detected cars at the intersection at the current timestep $i$, $c_{P_1}^{i-1}, c_{P_2}^{i-1}, c_{P_3}^{i-1}$ are the detected cars in the previous timestep $i-1$, and $q^{i-1}_1, q^{i-1}_2, q^{i-1}_3, q^{i-1}_4$ corresponds to Alice's position in the queue from the previous timestep.
The planner output is given as $O_{\text{Planner}} = \{q_1^i,q_2^i,q_3^i,q_4^i \}$, which corresponds to Alice's updated position in the queue for this timestep.
The Planner contract is given as
\begin{equation*}
\mc{C}_{\text{planner}} = \{I_{\text{planner}},O_{\text{planner}},\mathfrak{a}_{\text{planner}}, \mathfrak{g}_{\text{planner}} \},
\end{equation*}
where the assumptions are $\mathfrak{a}_{\text{planner}} = \{(q^{i-1}_1 \land \neg q^{i-1}_2 \land \neg q^{i-1}_3 \land \neg q^{i-1}_4) \lor  (\neg q^{i-1}_1 \land q^{i-1}_2 \land \neg q^{i-1}_3 \land \neg q^{i-1}_4) \lor (\neg q^{i-1}_1 \land \neg q^{i-1}_2 \land q^{i-1}_3 \land \neg q^{i-1}_4) \lor (\neg q^{i-1}_1 \land \neg q^{i-1}_2 \land \neg q^{i-1}_3 \land q^{i-1}_4)\}$, which ensures that Alice can only be in one position in the queue in the previous timestep.
The guarantees $\mathfrak{g}_{\text{planner}}$ describe how Alice's updated position in the queue will be determined.
\revise{If none of the other cars leave the intersection, Alice will stay in the same position in the queue,
\begin{equation*}
\begin{split}
\mathfrak{g}_0 = &\{( c_{P_1}^i \Leftrightarrow c_{P_1}^{i-1}) \land ( c_{P_2}^i \Leftrightarrow c_{P_2}^{i-1}) \land ( c_{P_3}^i \Leftrightarrow c_{P_3}^{i-1}) \land q_4^{i-1} \Rightarrow q_4^i,\\
&( c_{P_1}^i \Leftrightarrow c_{P_1}^{i-1}) \land ( c_{P_2}^i \Leftrightarrow c_{P_2}^{i-1}) \land ( c_{P_3}^i \Leftrightarrow c_{P_3}^{i-1}) \land q_3^{i-1} \Rightarrow q_3^i,\\
&( c_{P_1}^i \Leftrightarrow c_{P_1}^{i-1}) \land ( c_{P_2}^i \Leftrightarrow c_{P_2}^{i-1}) \land ( c_{P_3}^i \Leftrightarrow c_{P_3}^{i-1}) \land q_2^{i-1} \Rightarrow q_2^i,\\
&( c_{P_1}^i \Leftrightarrow c_{P_1}^{i-1}) \land ( c_{P_2}^i \Leftrightarrow c_{P_2}^{i-1}) \land ( c_{P_3}^i \Leftrightarrow c_{P_3}^{i-1}) \land q_1^{i-1} \Rightarrow q_1^i\}.
\end{split}
\end{equation*}}
\noindent
In the case of a single car leaving the intersection, Alice would advance by one in the queue, described as follows
\begin{equation*}
\begin{split}
\mathfrak{g}_1 =  \{ & ( c_{P_1}^i \Leftrightarrow c_{P_1}^{i-1}) \land ( \lnot c_{P_2}^i \land c_{P_2}^{i-1}) \land ( \lnot c_{P_3}^i \land c_{P_3}^{i-1}) \land q_4^{i-1}\Rightarrow q_3,\\
& ( c_{P_1}^i \Leftrightarrow c_{P_1}^{i-1}) \land ( \lnot c_{P_2}^i \land c_{P_2}^{i-1}) \land (  \lnot c_{P_3}^i \land c_{P_3}^{i-1}) \land q_3^{i-1}\Rightarrow q_2,\\
& ( c_{P_1}^i \Leftrightarrow c_{P_1}^{i-1}) \land ( \lnot c_{P_2}^i \land c_{P_2}^{i-1}) \land ( \lnot c_{P_3}^i \land c_{P_3}^{i-1}) \land q_2^{i-1}\Rightarrow q_1, \\
& (\lnot c_{P_1}^i \land c_{P_1}^{i-1}) \land ( c_{P_2}^i \Leftrightarrow c_{P_2}^{i-1}) \land ( \lnot c_{P_3}^i \land c_{P_3}^{i-1}) \land q_4^{i-1}\Rightarrow q_3,\\
& (\lnot c_{P_1}^i \land c_{P_1}^{i-1}) \land ( c_{P_2}^i \Leftrightarrow c_{P_2}^{i-1}) \land ( \lnot c_{P_3}^i \land c_{P_3}^{i-1}) \land q_3^{i-1}\Rightarrow q_2,\\
& (\lnot c_{P_1}^i \land c_{P_1}^{i-1}) \land ( c_{P_2}^i \Leftrightarrow c_{P_2}^{i-1}) \land ( \lnot c_{P_3}^i \land c_{P_3}^{i-1}) \land q_2^{i-1}\Rightarrow q_1, \\
& (\lnot c_{P_1}^i \land c_{P_1}^{i-1}) \land ( \lnot c_{P_2}^i \land c_{P_2}^{i-1}) \land ( c_{P_3}^i \Leftrightarrow c_{P_3}^{i-1}) \land q_4^{i-1}\Rightarrow q_3,\\
& (\lnot c_{P_1}^i \land c_{P_1}^{i-1}) \land ( \lnot c_{P_2}^i \land c_{P_2}^{i-1}) \land ( c_{P_3}^i \Leftrightarrow c_{P_3}^{i-1}) \land q_3^{i-1}\Rightarrow q_2,\\
& (\lnot c_{P_1}^i \land c_{P_1}^{i-1}) \land ( \lnot c_{P_2}^i \land c_{P_2}^{i-1}) \land ( c_{P_3}^i \Leftrightarrow c_{P_3}^{i-1}) \land q_2^{i-1}\Rightarrow q_1 \}.
\end{split} 
\end{equation*}

\noindent
If any two cars leave the intersection, Alice would advance in the queue by two steps, given by the following
\begin{equation*}
\begin{split}
\mathfrak{g}_2 = & \{(c_{P_1}^i \Leftrightarrow c_{P_1}^{i-1}) \land (\lnot c_{P_2}^i \land c_{P_2}^{i-1}) \land (\lnot c_{P_3}^i \land c_{P_3}^{i-1}) \land q_4^{i-1}\Rightarrow q_3,\\
& (c_{P_1}^i \Leftrightarrow c_{P_1}^{i-1}) \land (\lnot c_{P_2}^i \land c_{P_2}^{i-1}) \land (\lnot c_{P_3}^i \land c_{P_3}^{i-1}) \land q_3^{i-1}\Rightarrow q_1,\\
& (\lnot c_{P_1}^i \land c_{P_1}^{i-1}) \land ( c_{P_2}^i \Leftrightarrow c_{P_2}^{i-1}) \land (\lnot c_{P_3}^i \land c_{P_3}^{i-1}) \land q_4^{i-1}\Rightarrow q_3,\\
& (\lnot c_{P_1}^i \land c_{P_1}^{i-1}) \land ( c_{P_2}^i \Leftrightarrow c_{P_2}^{i-1}) \land (\lnot c_{P_3}^i \land c_{P_3}^{i-1}) \land q_3^{i-1}\Rightarrow q_1,\\
& (\lnot c_{P_1}^i \land c_{P_1}^{i-1}) \land ( \lnot c_{P_2}^i \land c_{P_2}^{i-1}) \land ( c_{P_3}^i \Leftrightarrow c_{P_3}^{i-1}) \land q_4^{i-1}\Rightarrow q_3,\\
& (\lnot c_{P_1}^i \land c_{P_1}^{i-1}) \land ( \lnot c_{P_2}^i \land c_{P_2}^{i-1}) \land ( c_{P_3}^i \Leftrightarrow c_{P_3}^{i-1}) \land q_3^{i-1}\Rightarrow q_1\}.
\end{split}
\end{equation*}
\noindent
If all three cars left the intersection in this timestep, Alice would move from $4^{th}$ position to the $1^{st}$ position in the queue, this is given by
\begin{equation*}
\begin{split}
\mathfrak{g}_3 = \{((\lnot c_{P_1}^i \land c_{P_1}^{i-1}) \land (\lnot c_{P_2}^i \land c_{P_2}^{i-1}) \land (\lnot c_{P_3}^i \land c_{P_3}^{i-1}) \land q_4^{i-1}\Rightarrow q_1) \}.
\end{split}
\end{equation*}
\noindent
Additionally, if Alice is in the first spot, it will stay in this spot until it takes its turn, $\mathfrak{g}_4 = \{q_1^{i-1} \Rightarrow q_1\}$.The Planner guarantees are thus given as 
\vspace{-2mm}
$$\mathfrak{g}_{\text{planner}} = \mathfrak{g}_0 \cup \mathfrak{g}_1 \cup \mathfrak{g}_2 \cup \mathfrak{g}_3 \cup \mathfrak{g}_4.$$
The Tracker component ensures that Alice will only move into the intersection if it has the right of way. The contract is given as
\begin{equation*}
\mc{C}_{\text{tracker}} = \{\{q_1^i, \text{icy\_roads}\},\{v^i \}, \mathfrak{a}_{\text{tracker}}, \mathfrak{g}_{\text{tracker}}  \},
\end{equation*}
with the input variables whether it is Alice's turn ($q_1^i$ provided by the Planner, and the output variable being Alice's speed $v^i$ (where $v$ evaluating to True corresponds to a positive speed).
The assumptions are $\mathfrak{a}_{Tracker} = \{ \neg \text{icy\_roads}\}$, and the guarantees are $\mathfrak{g} = \{q_1^i \Leftrightarrow v^i \}$. This ensures that only when $q_1^i$ is True, Alice will move into the intersection with a positive speed.
\revise{In addition to the above explained component assumptions and guarantees, the three components also contain additional assumptions and guarantees that are unrelated to the process of determining whether it is Alice's turn. These other guarantees represent different viewpoints of the components, and we will show that our framework accurately pinpoints the potential causes within the relevant terms outlined above. For each component, we added $100$ input and output variables, and $100$ component guarantees each, which result in $100$ system observations per timestep (tracker output variables), as shown in gray in Figure~\ref{fig:alice-diagram}. These variables and their interactions correspond to other component functionalities of Alice's system.}
\vspace{-8mm}
\revise{
\usetikzlibrary{positioning, arrows.meta, fit, calc}
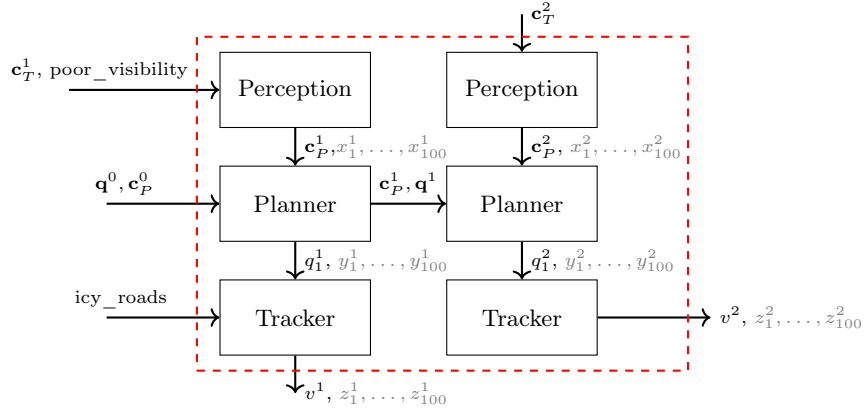
\begin{figure}[ht]
\centering
\begin{tikzpicture}[
    block/.style={draw, minimum width=2cm, minimum height=1cm},
    arrow/.style={->, thick},
    labelstyle/.style={font=\scriptsize},
    node distance=0.5cm and 1cm,
]

\node[block] (perc1) {Perception};
\node[block, below=of perc1] (plan1) {Planner};
\node[block, below=of plan1] (track1) {Tracker};

\node[block, right=of perc1] (perc2) {Perception};
\node[block, below=of perc2] (plan2) {Planner};
\node[block, below=of plan2] (track2) {Tracker};

\draw[arrow] (perc1) -- node[labelstyle, right] {$\mathbf{c}_P^1$,\textcolor{gray}{$x_1^1,\ldots, x_{100}^1$}} (plan1);
\draw[arrow] (plan1) -- node[labelstyle, right] {$q_1^1$, \textcolor{gray}{$y_1^1,\ldots, y_{100}^1$}} (track1);
\draw[arrow] (perc2) -- node[labelstyle, right] {$\mathbf{c}_P^2$, \textcolor{gray}{$x_1^2,\ldots, x_{100}^2$}} (plan2);
\draw[arrow] (plan2) -- node[labelstyle, right] {${q}_1^2$, \textcolor{gray}{$y_1^2,\ldots, y_{100}^2$}} (track2);

\draw[arrow] (plan1) -- node[labelstyle, above] {$\mathbf{c}_P^1, \mathbf{q}^1$} (plan2);

\draw[arrow] ([xshift=-2cm]perc1.west) -- node[labelstyle, above left, xshift=20pt] {$\mathbf{c}_T^1$, poor\_visibility} (perc1.west);
\draw[arrow] ([xshift=-1.5cm]plan1.west) -- node[labelstyle, above left] {$\mathbf{q}^0, \mathbf{c}_P^0$} (plan1.west);
\draw[arrow] ([xshift=-1.5cm]track1.west) -- node[labelstyle, above left, xshift=5pt] {icy\_roads} (track1.west);
\draw[arrow] ([yshift=0.5cm]perc2.north) -- node[labelstyle, anchor=south west] {$\mathbf{c}_T^2$} (perc2.north);

\draw[arrow] (track1.south) -- ++(0,-0.5) node[labelstyle, right] {$v^1$, \textcolor{gray}{$z_1^1,\ldots, z_{100}^1$}};
\draw[arrow] (track2.east) -- ++(1.5,0) node[labelstyle, right] {$v^2$, \textcolor{gray}{$z_1^2,\ldots, z_{100}^2$}};

\path let 
  \p1 = (perc1.north west), 
  \p2 = (track2.south east)
in
  coordinate (boxNW) at (\x1 - 0.3cm, \y1 + 0.2cm)
  coordinate (boxSE) at (\x2 + 1.2cm, \y2 - 0.2cm);

\node[
  draw=red, 
  dashed, 
  thick, 
  fit=(boxNW)(boxSE), 
  inner sep=0pt,
  overlay
] {};
\end{tikzpicture}
\vspace{-2mm}
\caption{Two-step execution of a perception–planner–tracker pipeline with intermediate signals. The bold variables denote vectors of the indicator variables.}
\label{fig:alice-diagram}
\vspace{-7mm}
\end{figure}
\noindent
Next, we compose Alice's system for one timestep by composing the contracts in their provided composition order $\mathtt{CompOrd} = [\mc{C}_{\text{perception}}, \mc{C}_{\text{planner}}, \mc{C}_{\text{tracker}}]$. We can then compose the system for a sequence of time steps, which allows us to evaluate the system behavior using the system-level variables. This composition is shown in Figure~\ref{fig:alice-diagram}. The resulting system contract comprises 212 guarantee terms, each involving a logical statement and some being lengthy and complex.
We are given the test trace as valuations of the variables for a series of time steps shown in Table~\ref{tab:alice_trace}. The trace contains the observations of the other cars in the intersection, and the visibility and road conditions. Alice's system-level output is its speed $v^i$ and the additional tracker output variables (not shown in the table). The internal variables are not accessible from the system level, except for the initial timestep as we require an initial condition as input for our contracts.
\vspace{-8mm}
\begin{table}[ht]
\centering
\caption{System-level violating trace (relevant variables)}
\begin{tabular}{c@{\hskip 1em}ccccccccccccc}
\toprule
\textbf{Time step $i$} & poor\_visibility & icy\_roads & $c_{T_1}^i$ & $c_{T_2}^i$ & $c_{T_3}^i$ & $c_{P_1}^i$ & $c_{P_2}^i$ & $c_{P_3}^i$ & $q_1^i$ & $q_2^i$ & $q_3^i$ & $q_4^i$ & $v^i$ \\
\midrule
0 & 0 & 0 & 1 & 1 & 1 & 1 & 1 & 1 & 0 & 0 & 0 & 1 & 0 \\
1 & 0 & 0 & 1 & 1 & 1 & -- & -- & -- & -- & -- & -- & -- & 0 \\
2 & 0 & 0 & 1 & 1 & 1 & -- & -- & -- & -- & -- & -- & -- & 1 \\
\bottomrule
\end{tabular}
\label{tab:alice_trace}
\vspace{-5mm}
\end{table}
\noindent
We can now analyze the given trace and we notice that one of the 212 system-level guarantees are violated. We use the diagnostics process to track this guarantee to the possible component causes, which requires us to check $50$ component-level statements (out of 656 total). Using this framework, we only have to check $7.6$\% of the component level terms. This now guides the the debugging process to access the relevant internal variables to evaluate the component-level statements and we will identify that the perception component did not provide its guarantees and did not detect the cars in time step $1$. For implementation details, please refer to the notebook in this repository\footnote{\url{https://github.com/pacti-org/cs-diagnostics/tree/main}}.

}

%% file: conclusion.tex
\vspace{-3mm}
\section{Conclusion and Future Work}
\vspace{-2mm}
In this paper, we proposed a diagnostics methodology based on assume-guarantee contracts using Pacti, a tool for compositional system analysis and design. We formally characterized when a system-level guarantee failure can be traced to a component, defined the required structure for contract composition, and identified the information that must be stored to enable diagnostics.

The diagnostics framework involves composing the system, generating a diagnostics map, and systematically checking component-level guarantees and assumptions using log data. We show that if the log contains valuations of the full system state, the method reliably identifies the responsible component. By focusing on the most relevant contract terms first, the approach often reduces the number of evaluated statements. In the worst case, all component-level contracts must be checked—but the framework prioritizes likely culprits and can significantly reduce effort in practice.

We demonstrated the methodology on abstract examples and two simplified scenarios inspired by real-world autonomous system tests. As future work, we plan to explore connections to behavior explainability in robotics. While our current method uses the contract composition operator, a similar approach could be developed using the contract quotient to identify missing components and trace the origin of contract terms, offering additional insight into system behavior.